\documentclass[a4paper,12pt]{article}
\usepackage[utf8]{inputenc}
\usepackage{amsmath, amsthm, amssymb, comment, hyperref, cite, graphicx, authblk}

\newtheorem{remark}{Remark}
\newtheorem{theorem}{Theorem}

\usepackage{anysize}
\marginsize{2.0cm}{2.0cm}{2.0cm}{2.0cm}
\begin{document}
\title{A General Representation for the Green's Function of Second Order Nonlinear Differential Equations}

\author[1]{Marco Frasca}

\author[2, 3]{Asatur Zh. Khurshudyan\footnote{Email: khurshudyan@mechins.sci.am}}

\affil[1]{\small Via Erasmo Gattamelata 3, 00176 Roma, Italy}
\affil[2]{\small Department on Dynamics of Deformable Systems and Coupled Fields, Institute of Mechanics, National Academy of Sciences of Armenia, Yerevan, Armenia}
\affil[3]{\small Institute of Natural Sciences, Shanghai Jiaotong University, Shanghai, China}

\date{}

\maketitle

\begin{abstract}
In this paper we study some classes of second order non-homogeneous nonlinear differential equations allowing a specific representation for nonlinear Green's function. In particular, we show that if the nonlinear term possesses a special multiplicativity property, then its Green's function is represented as the product of the Heaviside function and the general solution of the corresponding homogeneous equations subject to non-homogeneous Cauchy conditions. Hierarchies of specific non-linearities admitting this representation are derived. The nonlinear Green's function solution is numerically justified for the sinh-Gordon and Liouville equations. We also list two open problems leading to a more thorough characterizations of non-linearities admitting the obtained representation for the nonlinear Green's function.

\end{abstract}

\section{Introduction}

One of the most effective methods of analysis of \emph{linear} non-homogeneous ODEs and PDEs and their coupled systems is the well-known Green's function method \cite{Duffy2015}. If the Green's function of a differential equation, i.e., the general solution of the differential equation with the non-homogeneities substituted by the Dirac delta function, is known, then its general solution is represented in the form of the convolution of the Green's function and the non-homogeneities. Being based on the superposition principle, the Green's function method is believed to be applicable only to linear systems. However, an extension of the Green's function method to second order nonlinear differential equations has been developed by the first author of this paper alost a decade ago and applied in derivation of important results in quantum field theory. It has been hypothesized and \emph{numerically} justified in \cite{Frasca2006, Frasca2007, Frasca2008} that there exists a certain link between the following second order \emph{nonlinear} ordinary differential equations:
\begin{equation}\label{Frascaeq}
\frac{d^2 w}{d t^2} + N\left(w\right) = f\left(t\right), ~~ t > 0,
\end{equation}
and
\begin{equation}\label{Greeneq}
\frac{d^2 G}{d t^2} + N\left(G\right) = s \delta\left(t\right), ~~ t > 0,
\end{equation}
where $N$ is a generic nonlinear function, $f$ is a given function, $s$ is a scaling factor and $\delta$ is the Dirac distribution. That link has been made explicit in \cite{Frasca2008}. It turns out that the general solution of (\ref{Frascaeq}) is expressed in terms of the general solution of (\ref{Greeneq}) through the short time expansion
\begin{equation}\label{shorttex}
w\left(t\right) = \sum_{k = 0}^\infty \alpha_k \int_0^t \left(t - \tau\right)^k G\left(t - \tau\right) f\left(\tau\right) {\rm d}\tau,
\end{equation}
where $\alpha_k$, $k = 0, 1, 2, \dots,$ are the unknown expansion coefficients determined in terms of the quantities $w^{\left(k\right)}\left(0\right)$. Due to the similarity with the linear case, hereinafter, we will formally refer to the solution of (\ref{Greeneq}) as the nonlinear Green's function of (\ref{Frascaeq}).

It has been shown in \cite{Frasca2007} that the cubic and sine non-linearities
\[
N\left(w\right) = w^3
\]
and
\[
N\left(w\right) = \sin w,
\]
allow explicit determination of nonlinear Green's functions as follow:
\[
G\left(t\right) = 2^{\frac{1}{4}} \theta\left(t\right) \cdot \operatorname{sn}\left[\frac{t}{2^{\frac{1}{4}}}, i\right]
\]
and
\[
G\left(t\right) = 2\theta\left(t\right) \cdot \operatorname{am}\left[\frac{t}{\sqrt{2}},\sqrt{2}\right],
\]
respectively. Here $\theta$ is the Heaviside function
\[
\theta\left(t\right) = \begin{cases} 1, ~~ t > 0, \\ 0, ~~ t \leq 0, \end{cases}
\]
$\operatorname{sn}$ and $\operatorname{am}$ are the Jacobi snoidal and amplitude functions, respectively.

Moreover, it has been shown that the first term of (\ref{shorttex}) with $\alpha_0 = 1$ has been shown to provide an efficient numerical approximation for the solution of (\ref{Frascaeq}) for small $t$ \cite{Frasca2008, Khurshudyan2018, Khurshudyan2018AMP}.

Some new non-linearities admitting explicit and implicit determination of nonlinear Green's function have been considered in \cite{Khurshudyan2018, Khurshudyan2018AMP, Khurshudyan2018IJMPC}. Furthermore, it has been shown that the short time expansion (\ref{shorttex}) can be applied to represent the general solution of nonlinear PDEs in terms of their Green's function.

The main issue arising in consideration of Green's functions for nonlinear differential equations is the mathematical sense in which the solution of (\ref{Greeneq}) need to be considered. More precisely, because of the presence of the Dirac delta function in the right-hand side of (\ref{Greeneq}), it becomes obvious that its solution is a distribution, rather than a proper function. This means that the nonlinear term $N$ must be justified for distributions. Some aspects of the nonlinear theory of distributions can be found in \cite{Biagioni1990, Colombeau1992, Egorov1990, Ivanov2008}.

\section{Representation of the nonlinear Green's function in terms of the homogeneous solution}

For the non-linearities above, the nonlinear Green's function is determined as a product of the Heaviside function and Jacobi functions. This form strongly depends on the form of $N$. However, for instance, this does not hold for exponential non-linearity (see Section \ref{Liouville} below). As we show in this section, this representation holds for a specific class of non-linearities.

In order to state our main result, we introduce the following notations:
\[
\mathcal{D}_f = \left\{w : \mathbb{R}^+ \to \mathbb{R}; ~~ \frac{d^2 w}{d t^2} + N\left(w\right) = f\left(t\right), ~~ t > 0, ~~ w\left(0\right) = \frac{d w}{d t}\bigg|_{t = 0} = 0 \right\},
\]
\[
\mathcal{G}_s = \left\{G : \mathbb{R}^+ \to \mathbb{R}; ~~ \frac{d^2 G}{d t^2} + N\left(G\right) = s \delta\left(t\right), ~~ t > 0, ~~ G\left(0\right) = \frac{d G}{d t}\bigg|_{t = 0} = 0 \right\},
\]
\[
\mathcal{H}_s = \left\{w_0 : \mathbb{R}^+ \to \mathbb{R}; \frac{d^2 w_0}{d t^2} + N\left(w_0\right) = 0, ~~ t > 0, ~~ w_0\left(0\right) = 0, ~ \frac{d w_0}{d t}\bigg|_{t = 0} = s \right\}.
\]
We also introduce the following set of functions possessing the generalized multiplicability property:
\[
\mathcal{N} = \left\{ N : \mathbb{R} \to \mathbb{R}; ~~ N\left(\theta \cdot w\right) = \theta\left(t\right) \cdot N\left(w\right) \right\}.
\]
In this terminology, the short time expansion (\ref{shorttex}) provides a link between the elements of $\mathcal{G}_s$ and $\mathcal{D}_f$. The main result of this paper, stated in the following theorem, establishes a link between $\mathcal{G}_s$ and $\mathcal{H}_s$ and, thus, relates the solutions of a non-homogeneous equation and its homogeneous counterpart.

\begin{theorem}
If $N \in \mathcal{N}$, then any element of $\mathcal{G}_s$ admits the following representation:
\begin{equation}\label{Greenrep}
G\left(t\right) = \theta\left(t\right) w_0\left(t\right),
\end{equation}
where $w_0 \in \mathcal{H}_s$.
\end{theorem}

\begin{proof}
This is shown by a direct substitution of (\ref{Greenrep}) into the homogeneous equation
\begin{equation}\label{homogeq}
\theta\left(t\right) \cdot \frac{d^2 w_0}{d t^2} + \theta\left(t\right) \cdot N\left(w_0\right) = 0
\end{equation}
which holds for any element $w_0 \in \mathcal{H}_s$. Then, since for any $N \in \mathcal{N}$,
\[
\theta\left(t\right) \cdot N\left(w_0\right) = N\left(\theta \cdot w_0\right) = N\left(G\right),
\]
and for any $w_0 \in \mathcal{H}_s$,
\[
\frac{d^2 \left(\theta \cdot w_0\right)}{d t^2} = w_0\left(0\right) \delta'\left(t\right) + \frac{d w_0}{d t}\bigg|_{t = 0} \delta\left(t\right) + \theta\left(t\right) \cdot \frac{d^2 w_0}{d t^2} = \theta\left(t\right) \cdot \frac{d^2 w_0}{d t^2} + s \delta\left(t\right)
\]
(in the sense of distributions), then (\ref{homogeq}) is reduced to
\[
\frac{d^2 G}{d t^2} + N\left(G\right) = s \delta\left(t\right).
\]
Taking into account that in the sense of distributions
\[
\frac{d G}{d t} = w_0\left(0\right) \delta\left(t\right) + \theta\left(t\right) \frac{d w_0}{d t},
\]
and using the above definition of the Heaviside function, we see that $G$ satisfies homogeneous Cauchy conditions. Thus, any $G$ defined by (\ref{Greenrep}), belongs to $\mathcal{G}_s$.
\end{proof}

\begin{remark}
Note that (\ref{Greenrep}) has theoretical and computational advantages. While the method of the nonlinear theory of distributions must be employed for the characterization of the elements of $\mathcal{G}_s$, the characterization of the elements of $\mathcal{H}_s$ is much simpler. On the other hand, it is less costly to approximate $w_0 \in \mathcal{H}_s$, than $G \in \mathcal{G}_s$.
\end{remark}

\subsection{Particular forms of $N \in \mathcal{N}$}

In this section we find some families of $N$ possessing the multiplicability property
\begin{equation}\label{multyprop}
N\left(\theta \cdot w_0\right) = \theta\left(t\right) \cdot N\left(w_0\right).
\end{equation}
The following theorem lists some particular forms of $N$ satisfying (\ref{multyprop}).

\begin{theorem}
All the functions below possesses the generalized multiplicability property (\ref{multyprop}):
\begin{equation}\label{powernonlin}
N\left(w\right) = w^n, ~~ n \in \mathbb{N},
\end{equation}
\begin{equation}\label{sinhnonlin}
N\left(w\right) = \sinh w,
\end{equation}
\begin{equation}\label{tanhnonlin}
N\left(w\right) = \tanh w,
\end{equation}
\begin{equation}\label{sinnonlin}
N\left(w\right) = \sin w,
\end{equation}
\begin{equation}\label{tannonlin}
N\left(w\right) = \tan w,
\end{equation}
\begin{equation}\label{arcsinnonlin}
N\left(w\right) = \arcsin w,
\end{equation}
\begin{equation}\label{arctannonlin}
N\left(w\right) = \arctan w,
\end{equation}
\begin{equation}\label{arcsinhnonlin}
N\left(w\right) = \operatorname{arcsinh} w,
\end{equation}
\begin{equation}\label{arctanhnonlin}
N\left(w\right) = \operatorname{arctanh} w,
\end{equation}
\begin{equation}\label{lognonlin}
N\left(w\right) = \ln\left(1 + w\right).
\end{equation}
\end{theorem}

\begin{proof}
Here we carry out the proof for (\ref{powernonlin}) and (\ref{sinhnonlin}), since for (\ref{tanhnonlin})--(\ref{lognonlin}) it is carried put in the same way as for (\ref{sinhnonlin}).

In order to show that
\begin{equation}\label{thetapower}
\left[ \theta \cdot w \right]^n\left(t\right) = \theta\left(t\right) \cdot w^n\left(t\right),
\end{equation}
it suffices to show that
\begin{equation}\label{thetan}
\theta^n\left(t\right) = \theta\left(t\right).
\end{equation}
To this aim, we use the representation
\[
\theta\left(t\right) = \frac{1}{2} \left[1 + \operatorname{sign}\left(t\right) \right],
\]
where
\[
\operatorname{sign}\left(t\right) = \frac{\left|t\right|}{t}, ~~ t \in \mathbb{R},
\]
is the sign function. Obviously, according to the definition,
\begin{equation}\label{signn}
\operatorname{sign}^n\left(t\right) = \begin{cases} 1, ~~~~~~~~~~ n ~ {\rm is ~ even}, \\ \operatorname{sign}\left(t\right), ~~ n ~ {\rm is ~ odd}. \end{cases}
\end{equation}
Then, using the Newton's binomial formula, we have
\[
\theta^n\left(t\right) = \frac{1}{2^n} \sum_{k = 0}^n \frac{n!}{k! \left(n - k\right)!} \operatorname{sign}^k\left(t\right),
\]
which, in view of (\ref{signn}), can be transformed to the following:
\[\begin{split}
\sum_{k = 0}^n \frac{n!}{k! \left(n - k\right)!} \operatorname{sign}^k\left(t\right) &= \sum_{{\rm even} ~ k = 0}^n \frac{n!}{k! \left(n - k\right)!} + \sum_{{\rm odd} ~ k = 0}^n \frac{n!}{k! \left(n - k\right)!} \cdot \operatorname{sign}\left(t\right) = \\
&= 2^{n - 1} \left[ 1 + \operatorname{sign}\left(t\right) \right].
\end{split}\]
Therefore, we eventually arrive at
\[
\theta^n\left(t\right) = \frac{1}{2} \left[1 + \operatorname{sign}\left(t\right) \right] = \theta\left(t\right),
\]
which implies (\ref{powernonlin}).

\begin{remark}
A direct consequence of (\ref{thetan}) is the following relation:
\[
\frac{d \theta^n}{d t} = n \theta^{n - 1}\left(t\right) \frac{d \theta}{d t} = \delta\left(t\right).
\]
We find the last equality also in \cite{Egorov1990}. Thus, we reduce the constraint
\[
n \theta^{n - 1}\left(0\right) = 1.
\]
\end{remark}

Next, in order to show that
\[
\sinh\left[ \theta\left(t\right) w\left(t\right) \right] = \theta\left(t\right) \sinh\left[w\left(t\right)\right],
\]
we use the series expansion
\[
\sinh x = \sum_{n = 0}^\infty \frac{x^{2n + 1}}{\left(2 n + 1\right)!}.
\]
Then, by virtue of (\ref{thetapower}), we have
\[
\sinh\left[ \theta\left(t\right) w\left(t\right) \right] = \sum_{n = 0}^\infty \frac{\left[ \theta\left(t\right) w\left(t\right) \right]^{2n + 1}}{\left(2 n + 1\right)!} = \theta\left(t\right) \cdot \sum_{n = 0}^\infty \frac{w^{2n + 1}\left(t\right)}{\left(2 n + 1\right)!} = \theta\left(t\right) \sinh\left[w\left(t\right)\right].
\]

\end{proof}

\begin{remark}
Apparently, in all cases above, the multiplicative property (\ref{multyprop}) is not affected if we add to $N$ the linear term $g w_0$, with a given $g$. Similarly, the linear combination of all the particular functions above also possesses (\ref{multyprop}). Moreover, the product of any number of functions (\ref{powernonlin})--(\ref{lognonlin}) of any power $n \in \mathbb{N}$, also possesses (\ref{multyprop}).
\end{remark}

\subsection{Hierarchies of related nonlinear ODEs}

Using the particular solutions of (\ref{multyprop}) derived in the previous section, we provide some hierarchies of second order nonlinear equations allowing the representation (\ref{Greenrep}).

The relation (\ref{thetapower}) provides that (\ref{Greenrep}) is valid for the following sets equations:
\[
\frac{d^2 w}{d t^2} \pm w^n = f\left(t\right), ~~ n \in \mathbb{N},
\]
\[
\frac{d^2 w}{d t^2} \pm w^n \pm w^m = f\left(t\right), ~~ n, m \in \mathbb{N},
\]
\[
\frac{d^2 w}{d t^2} \pm \left( w^n \pm w^m \right)^k = f\left(t\right), ~~ n, m, k \in \mathbb{N}.
\]

The relations (\ref{sinhnonlin})--(\ref{lognonlin}) provide that (\ref{Greenrep}) is valid for the following sets equations:
\[
\frac{d^2 w}{d t^2} \pm \left( \sinh^n w \pm \sinh^m w \right)^k = f\left(t\right), ~~ n, m, k \in \mathbb{N},
\]
\[
\frac{d^2 w}{d t^2} \pm \left( \sinh^n w \pm \tanh^m w \right)^k = f\left(t\right), ~~ n, m, k \in \mathbb{N},
\]
\[
\frac{d^2 w}{d t^2} \pm \left( \tanh^n w \pm \tanh^m w \right)^k = f\left(t\right), ~~ n, m, k \in \mathbb{N},
\]
\[
\frac{d^2 w}{d t^2} \pm \sinh^n w \cdot \tanh^m w = f\left(t\right), ~~ n, m \in \mathbb{N}, ~~ n \geq m,
\]
\[
\frac{d^2 w}{d t^2} \pm \frac{\sinh^n w}{\tanh^m w} = f\left(t\right), ~~ n, m \in \mathbb{N}, ~~ n \geq m,
\]
\[
\frac{d^2 w}{d t^2} \pm \sinh^n w \pm w^m = f\left(t\right), ~~ n, m \in \mathbb{N},
\]
\[
\frac{d^2 w}{d t^2} \pm \tanh^n w \pm w^m = f\left(t\right), ~~ n, m \in \mathbb{N},
\]
\[
\frac{d^2 w}{d t^2} \pm \sinh^n w \cdot w^m = f\left(t\right), ~~ n, m \in \mathbb{N},
\]
\[
\frac{d^2 w}{d t^2} \pm \tanh^n w \cdot w^m = f\left(t\right), ~~ n, m \in \mathbb{N},
\]
\[
\frac{d^2 w}{d t^2} \pm \frac{\sinh^n w}{w^m} = f\left(t\right), ~~ n, m \in \mathbb{N}, ~~ n \geq m,
\]
\[
\frac{d^2 w}{d t^2} \pm \frac{\ln^n w}{\tanh^m w} \pm \sinh^k w \pm \left( \sin^p w^q \pm \tanh^r w^s \right)^l = f\left(t\right), ~~ k, l, n, m, p, q, r, s \in \mathbb{N}, ~~ n \geq m,
\]
and other combinations of power, sine, tangent, hyperbolic sine, hyperbolic tangent, inverse sine, inverse tangent, logarithmic functions.

\begin{remark}
The representation (\ref{Greenrep}) can be established for nonlinear PDEs that are reduced to second order nonlinear ODEs of the form (\ref{Frascaeq}) using, for instance, the traveling wave ansatz, the generalized separation of variables \cite{Khurshudyan2018}, etc.
\end{remark}

\section{The case of Liouville equation}\label{Liouville}

Apparently, the case $N\left(w\right) = \exp w$ of the Liouville equation \cite{Teschner:2001rv} seems to be at odd with the preceding discussion. Indeed, one can incur in a difficulty like
\begin{equation}\label{expnonlin}
\exp\left[ \theta\left(t\right) w_0\left(t\right) \right] = \sum_{n = 0}^\infty \frac{\left[ \theta\left(t\right) w_0\left(t\right) \right]^n}{n!} = 1 + \theta\left(t\right) \cdot \sum_{n = 1}^\infty \frac{w_0^n\left(t\right)}{n!} = 1 + \theta\left(t\right) \cdot \exp\left[ w_0\left(t\right) \right],
\end{equation}
and a direct solution does not seem to be available. The reason is that in this case there are two contributions to the solution. This can be written as
\begin{equation}\label{expnonlinsol}
G\left(t\right) = 2 \theta\left(t\right) \ln\left[\frac{\sqrt{2\epsilon}}{\cosh\left(\sqrt{\epsilon} t + \phi\right)}\right] + 2 \theta\left(-t\right) \ln\left[\frac{\sqrt{2\epsilon}}{\cosh\left(\sqrt{\epsilon} t - \phi\right)}\right]
\end{equation}
being $\epsilon$ and $\phi$ two constants fixed in such a way that
\begin{equation}
    -\tanh \phi = \frac{1}{4\sqrt{\epsilon}}.
\end{equation}

Consequences of this kind of solution could be that 2D quantum gravity is not confining but this would take us too far away from the content of the paper.

\section{Numerical analysis}

In this section we are going to show how the representation (\ref{Greenrep}) is applied to study particular nonlinear equations. Consider the ODE
\begin{equation}\label{nonhomopart}
\frac{d^2 w}{d t^2} + \sinh w = f\left(t\right).
\end{equation}
The general solution of the homogeneous equation
\[
\frac{d^2 w_0}{d t^2} + \sinh w_0 = 0
\]
subject to the following Cauchy conditions:
\[
w_0\left(0\right) = 0, ~~ \frac{d w_0}{d t}\bigg|_{t = 0} = s,
\]
reads as follows:
\[
w_0\left(t\right) = -2 i \operatorname{am}\left[\frac{i s}{2} t, -\frac{4}{s^2}\right],
\]
where $\operatorname{am}$ is the Jacobi amplitude function. Therefore, the nonlinear Green's function of the non-homogeneous equation (\ref{nonhomopart}) reads as follows:
\[
G\left(t\right) = -2i \theta\left(t\right) \cdot \operatorname{am}\left[\frac{i s}{2} t, -\frac{4}{s^2}\right].
\]
Thus, the general solution of (\ref{nonhomopart}) reads as follows:
\[
w\left(t\right) = -2i \sum_{k = 0}^\infty \alpha_k \int_0^t \tau^k \theta\left(\tau\right)  \operatorname{am}\left[\frac{i s}{2} \tau, -\frac{4}{s^2}\right] f\left(t - \tau\right) {\rm d}\tau.
\]

Now let us evaluate the numerical error of approximation of (\ref{nonhomopart}) by the finite sum
\[
w_K\left(t\right) \approx -2i \sum_{k = 0}^K \alpha_k \int_0^t \tau^k \theta\left(\tau\right)  \operatorname{am}\left[\frac{i s}{2} \tau, -\frac{4}{s^2}\right] f\left(t - \tau\right) {\rm d}\tau,
\]
when $t \in \left[0, 1\right]$. To this aim, we quantify the logarithmic error
\[
\operatorname{Er}\left(K; t\right) = \ln\left| w_K\left(t\right) - w_{\rm MoL}\left(t\right) \right|,
\]
where $w_{\rm MoL}$ is the numerical solution of (\ref{nonhomopart}) with the method of lines \cite{Schiesser1991}.

Figure \ref{fig1} below shows the evolution of $\operatorname{Er}$ when $K$ increases from $1$ to $4$, in the particular case when $s = 1$ and $f\left(t\right) = \delta\left(t\right)$.

\begin{figure}[h!]
\centerline{\includegraphics[width = 3.2in]{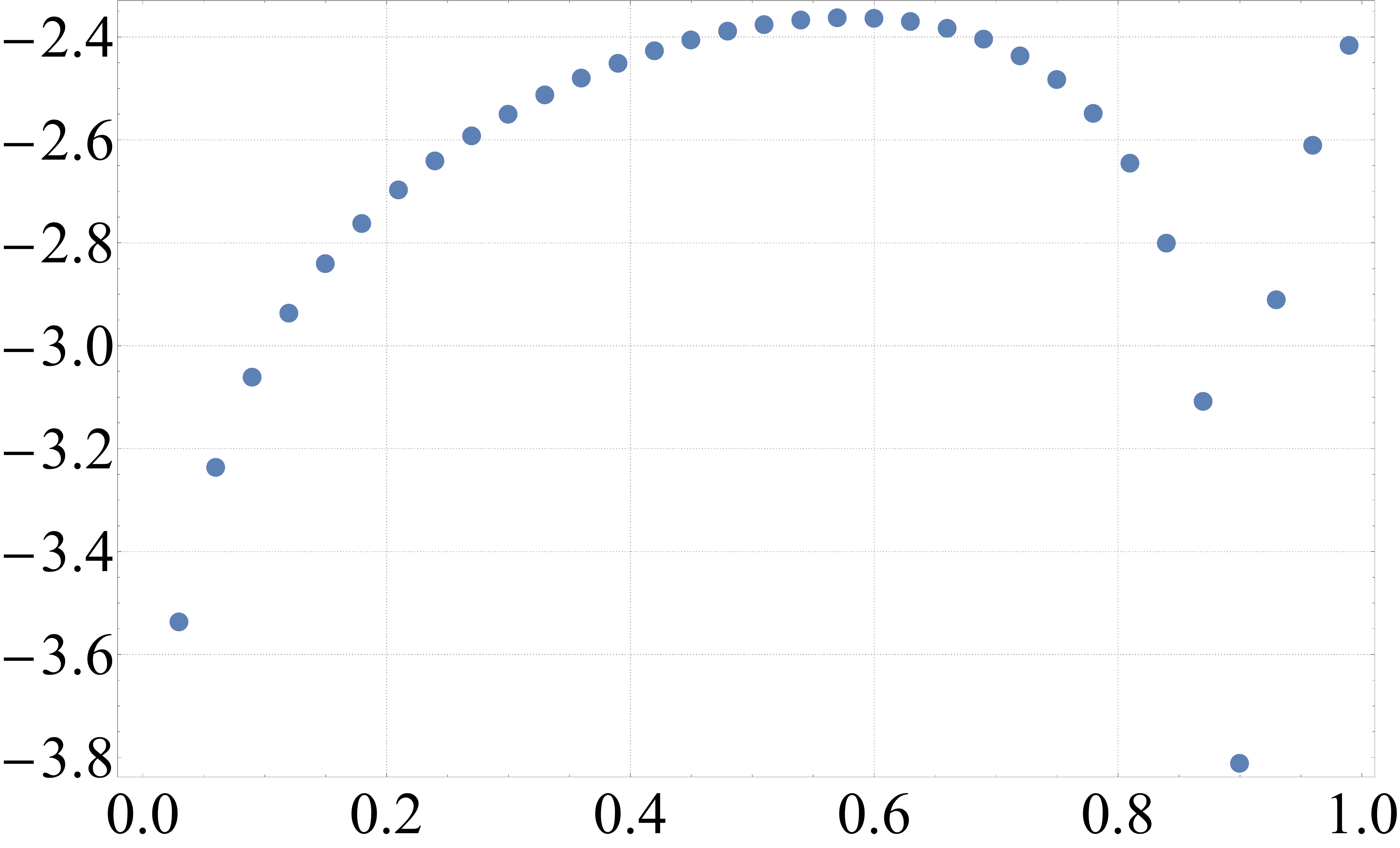} ~ \includegraphics[width=3.2in]{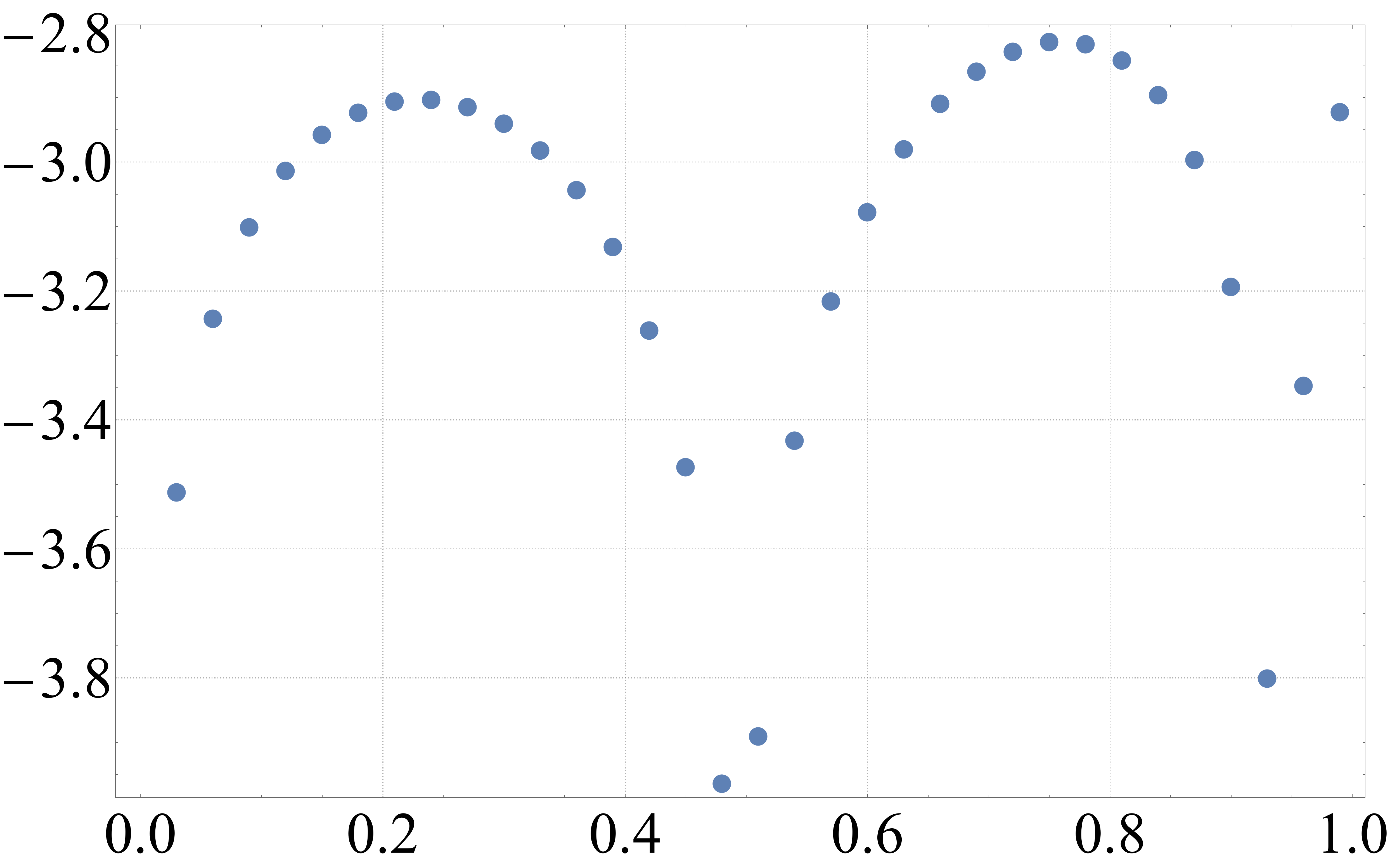}}
\bigskip
\centerline{\includegraphics[width = 3.2in]{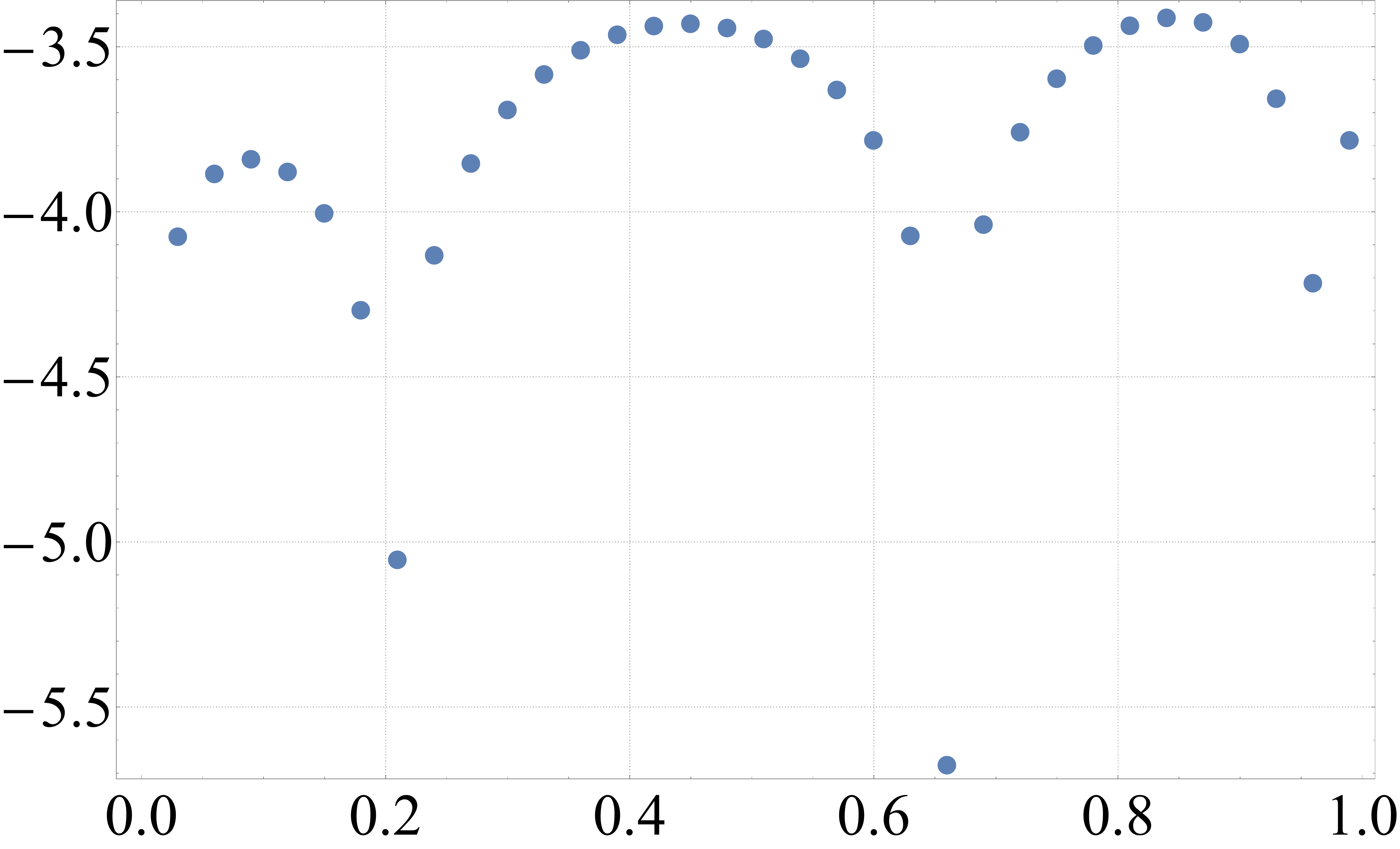} ~ \includegraphics[width=3.2in]{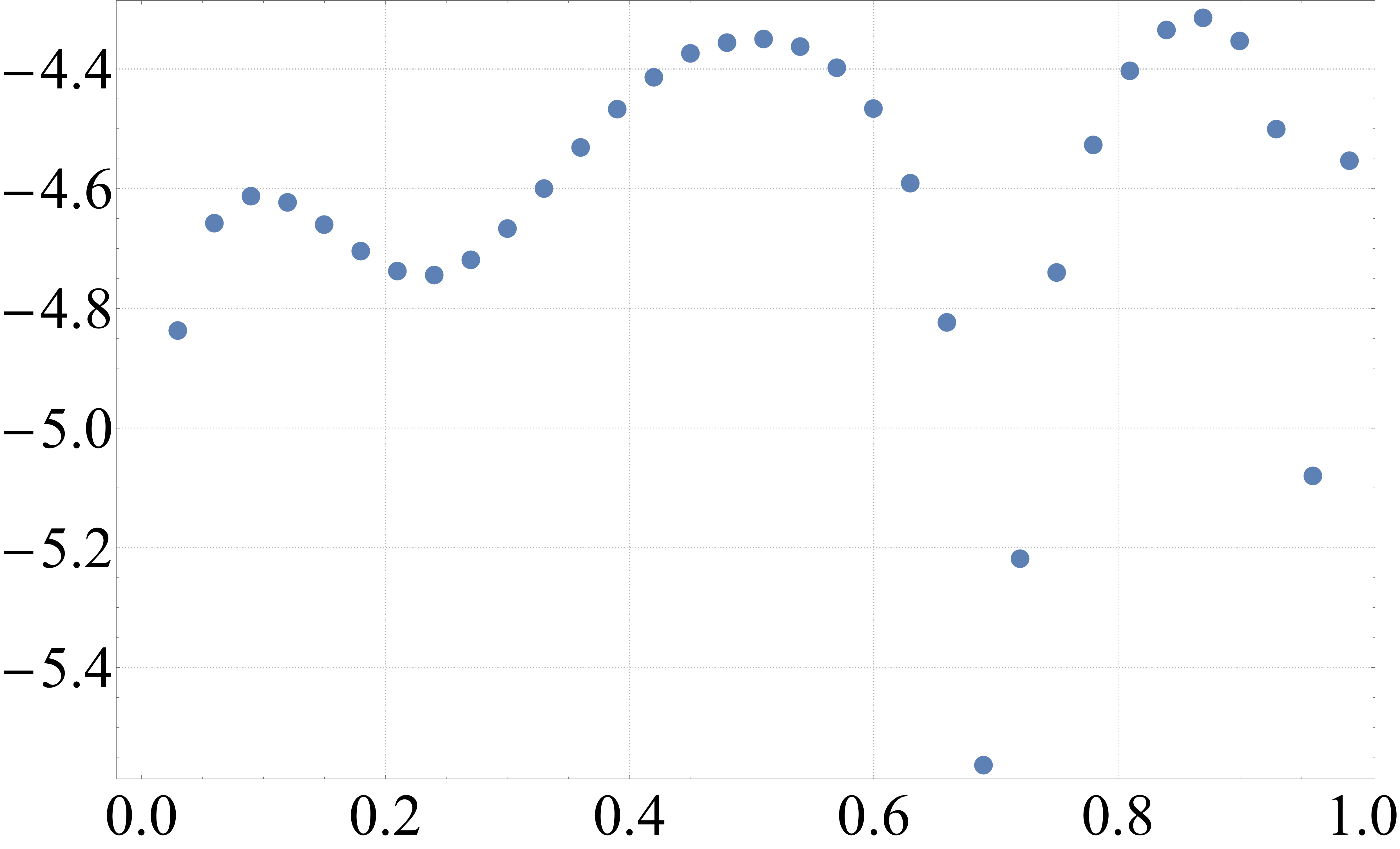}}

\caption{Discrete plots of $\operatorname{Er}\left(1; t\right)$, $\operatorname{Er}\left(2; t\right)$, $\operatorname{Er}\left(3; t\right)$, and $\operatorname{Er}\left(4; t\right)$ (respectively) against $t \in \left[0, 1\right]$ for $f\left(t\right) = \delta\left(t\right)$: (\ref{nonhomopart})}
\label{fig1}
\end{figure}

We also aim to show that the nonlinear Green's function of the Liouville equation derived in the previous section:
\[
\frac{d^2 w}{d t^2} + \exp w = f\left(t\right).
\]
As it is seen from Fig. \ref{fig2}, expressing the evolution of $\operatorname{Er}$ when $K$ increases from $1$ to $4$, higher order terms of the short time expansion lead to significant improvement of the approximation.

\begin{figure}[h!]
\centerline{\includegraphics[width = 3.2in]{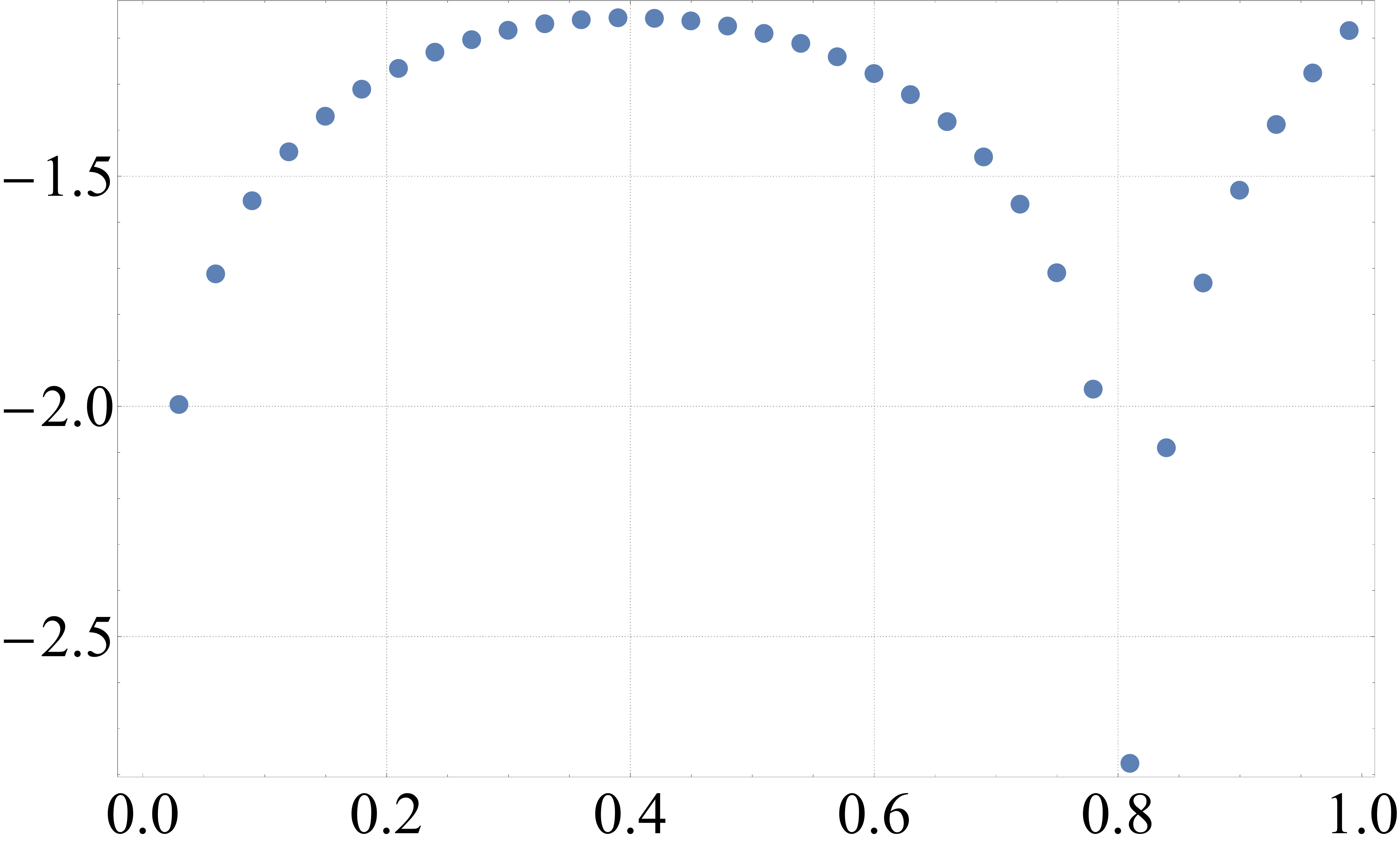} ~ \includegraphics[width=3.2in]{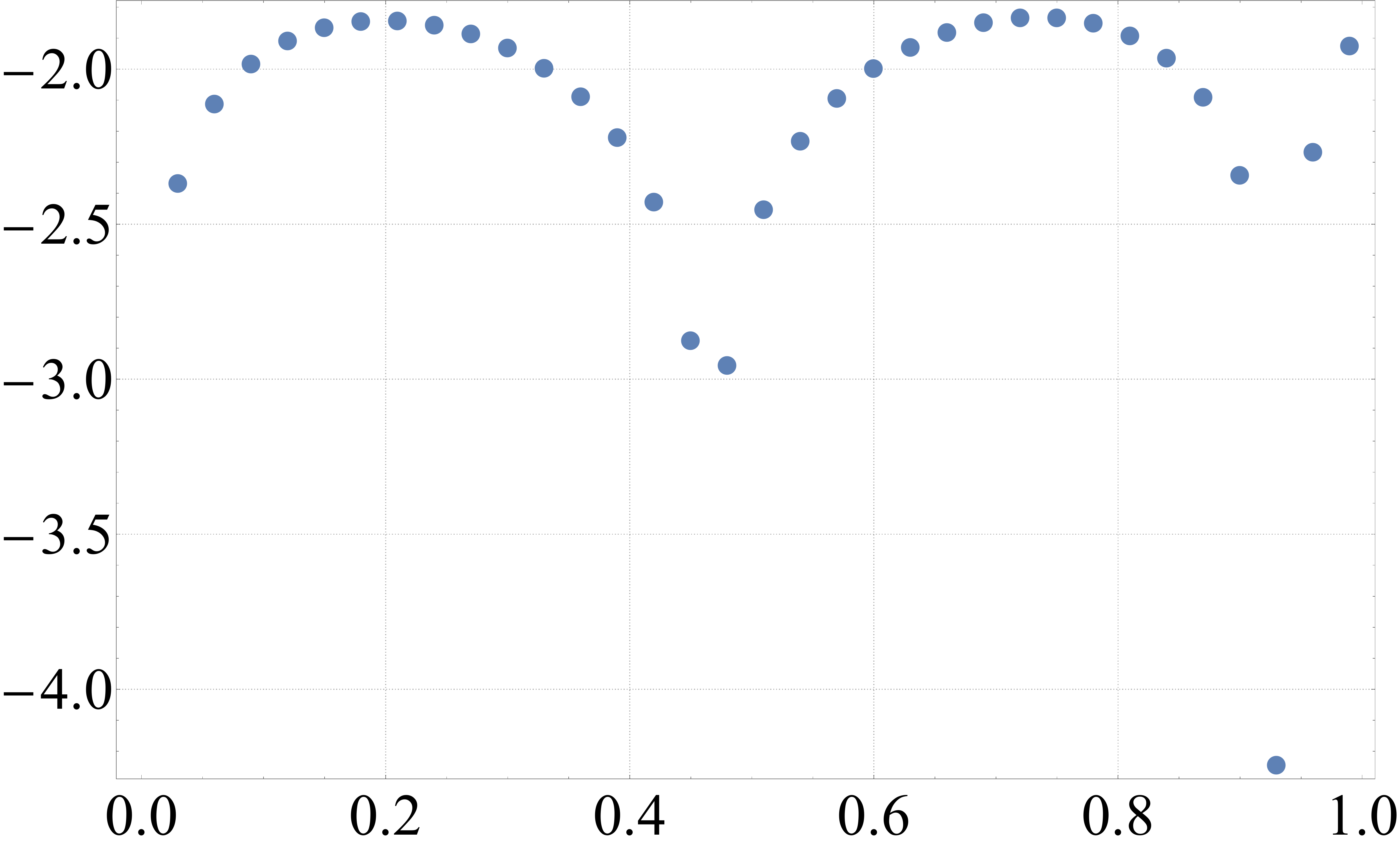}}
\bigskip
\centerline{\includegraphics[width = 3.2in]{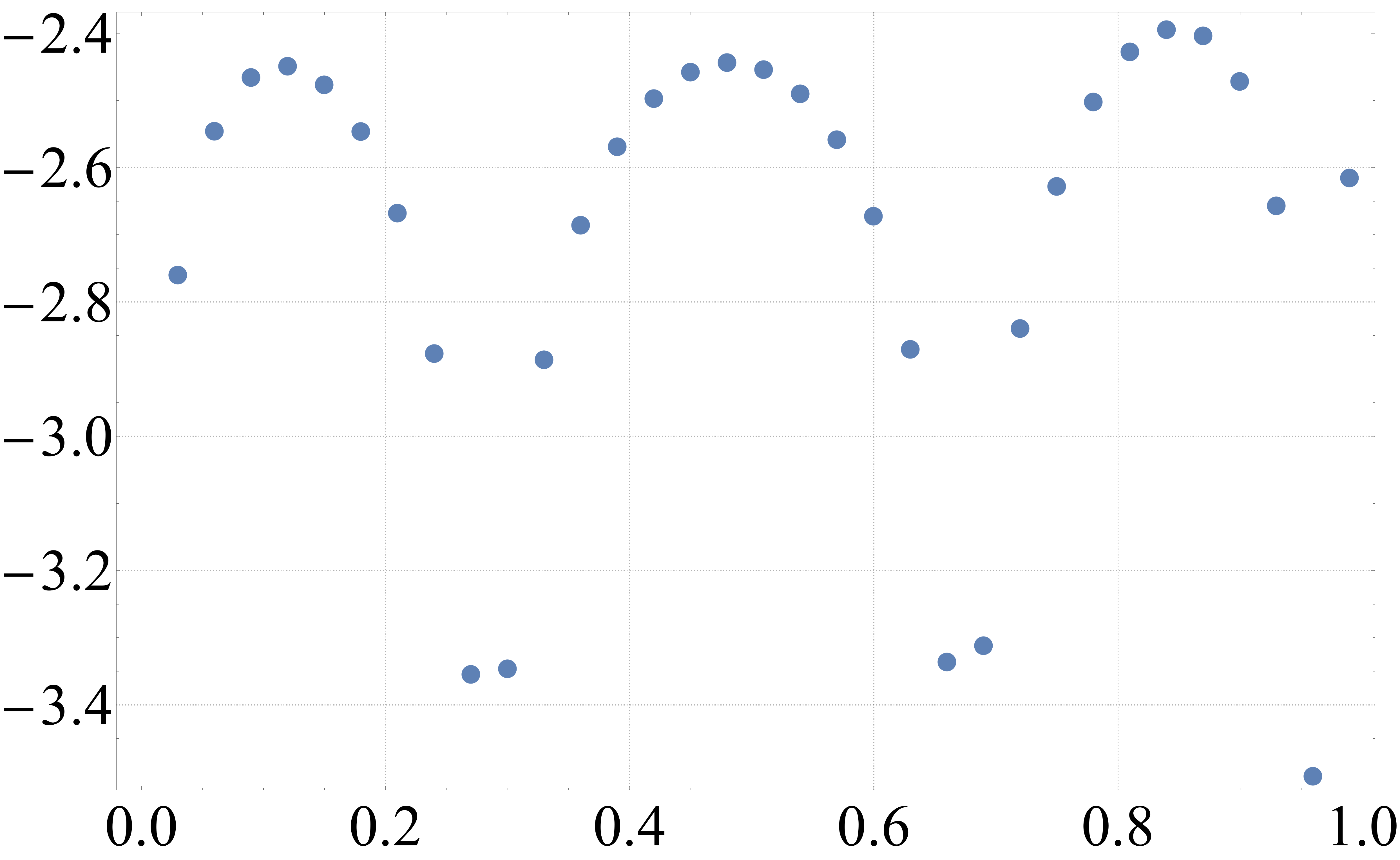} ~ \includegraphics[width=3.2in]{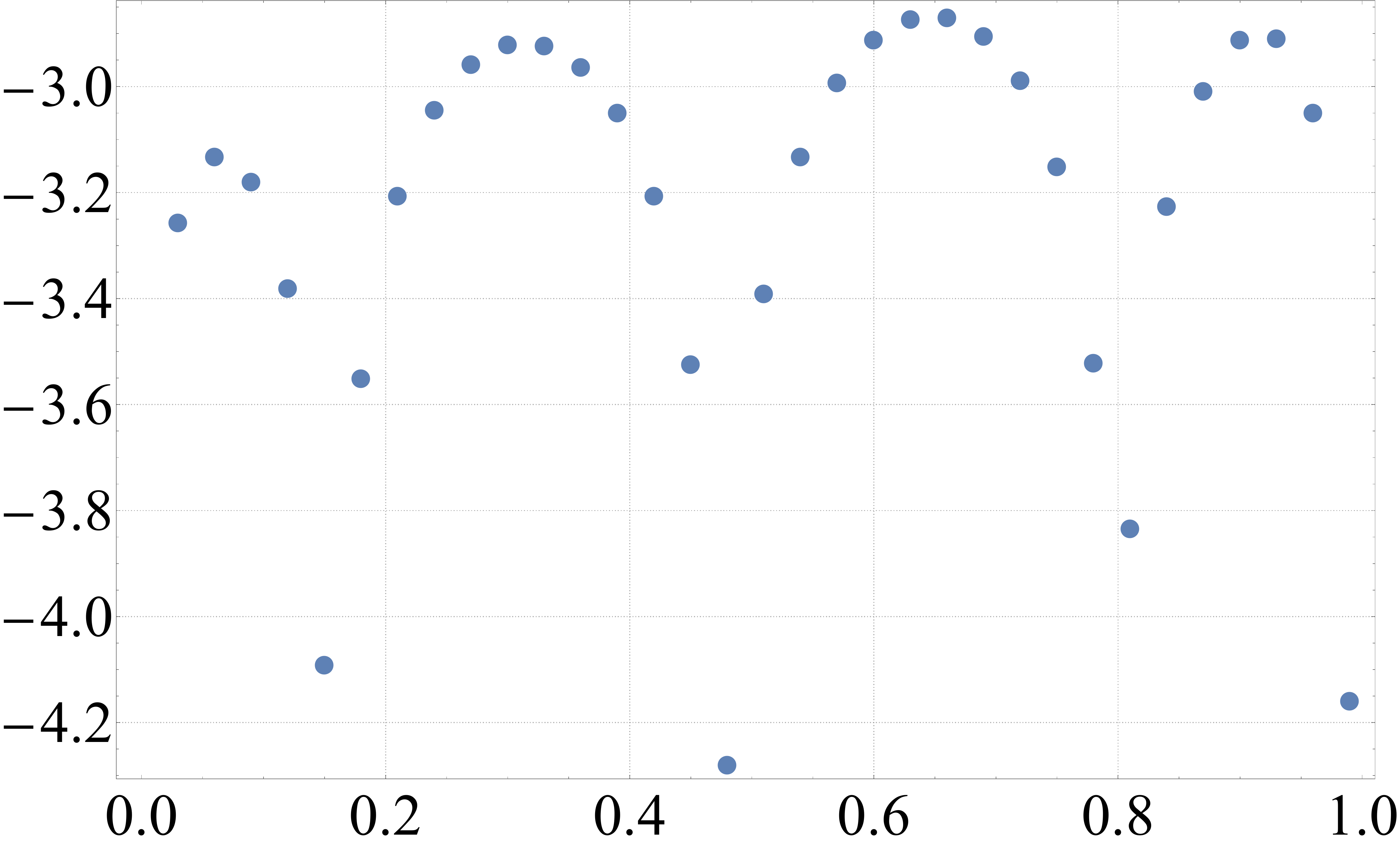}}

\caption{Discrete plots of $\operatorname{Er}\left(1; t\right)$, $\operatorname{Er}\left(2; t\right)$, $\operatorname{Er}\left(3; t\right)$, and $\operatorname{Er}\left(4; t\right)$ (respectively) against $t \in \left[0, 1\right]$ for $f\left(t\right) = \delta\left(t\right)$: Liouville equation}
\label{fig2}
\end{figure}

The logarithmic error decreases significantly with increase of $K$ also for other forms of $f$.

\section{Two open problems}

In this section we sum up some of the related open problems that we came across when studying particular solutions of (\ref{multyprop}).

\begin{enumerate}

\item The non-linearity (\ref{powernonlin}) is studied for $n \in \mathbb{N}$. However, numerical simulations show that (\ref{powernonlin}) satisfies (\ref{multyprop}) for any positive real $n$. This fact remains rigorously unproved.

\item The same concerns given in Section \ref{Liouville} seem to hold also for $\cos$, $\cosh$, $\cot$, $\coth$, $\arccos$, arccot, arccosh, arccoth functions. Further work is needed here as in the case of Liouville equation above.

\end{enumerate}

\section*{Conclusions}

We derive an equality type constraint for the nonlinear term of second order nonlinear ODEs providing an important representation formula for the nonlinear Green's function. More specifically, if the constraint holds, then the nonlinear Green's function is represented in the form of product of the Heaviside function and the general solution of the corresponding homogeneous equation. We also establish some particular solutions to the constraint allowing to represent the nonlinear Green's function of general hierarchies of ODEs in the mentioned form. The numerical comparison of the Green's function solution and the solution derived by the well-known numerical method of lines supports the efficiency of the result.

We also outline some open problems that will make a significant progress in the future.

\end{document}